\documentclass{article}
\def\be#1\ee{\begin{eqnarray}#1\end{eqnarray}}

\usepackage[margin=3.2cm]{geometry}

\usepackage{setspace}

\usepackage{graphics}

\usepackage{docmute}

\usepackage[utf8]{inputenc}

\usepackage{mathtools}

\usepackage{amsmath}

\usepackage{amsfonts}

\usepackage{amssymb}

\usepackage{amsthm}
\usepackage{aliascnt}

\usepackage{makeidx}
\makeindex

\usepackage{textcomp}
\usepackage[sb]{libertine}
\usepackage[varqu,varl]{zi4}%
\usepackage[libertine,bigdelims,vvarbb]{newtxmath}
\usepackage{superiors}
\usepackage{bm}
\useosf

\usepackage[breaklinks,linktocpage]{hyperref}
\usepackage[hyperpageref]{backref}

\usepackage{url}
\usepackage{breakurl}

\usepackage{tikz-cd}

\usepackage{enumitem}
\usepackage[UKenglish]{babel}
\usepackage[makeroom]{cancel}
\usepackage{pstricks-add}

\definecolor{shadecolor}{rgb}{0.88,0.91,0.95}       
\usepackage{tcolorbox}

\usepackage{dynkin-diagrams}
\usetikzlibrary{backgrounds}

\usepackage{booktabs}

\usepackage{fancyhdr}
\newcommand{\R}{\mathbb{R}}

\renewcommand{\Im}{\operatorname{Im}}

\newcommand{\SO}{\operatorname{SO}}

\newcommand{\tr}{\operatorname{tr}}

\newcommand{\Id}{\operatorname{Id}}

\renewcommand{\|}[1]{\left| \left| #1 \right| \right|}

\renewcommand{\d}{\operatorname{d} \!}

\renewcommand{\tilde}[1]{\widetilde{#1}}

\newcommand\Item[1][]{%
  \ifx\relax#1\relax  \item \else \item[#1] \fi
  \abovedisplayskip=0pt\abovedisplayshortskip=0pt~\vspace*{-\baselineskip}}

\usepackage[capitalise]{cleveref}

\numberwithin{equation}{section}

\newtheorem{proposition}[equation]{Proposition}
\crefname{proposition}{Proposition}{Propositions}

\crefname{lemma}{Lemma}{Lemmas}

\newtheorem{corollary}[equation]{Corollary}
\crefname{corollary}{Corollary}{Corollaries}

\newtheorem*{corollary*}{Corollary}
\crefname{corollary*}{Corollary}{Corollaries}

\newtheorem{theorem}[equation]{Theorem}
\crefname{theorem}{Theorem}{Theorems}

\crefname{conjecture}{Conjecture}{Conjectures}

\newtheorem*{theorem*}{Theorem}
\crefname{theorem*}{Theorem}{Theorems}

\crefname{claim}{Claim}{Claims}

\theoremstyle{remark}

\crefname{question}{Question}{Questions}

\newtheorem{definition}[equation]{Definition}
\crefname{definition}{Definition}{Definitions}

\crefname{example}{Example}{Examples}

\newtheorem{remark}[equation]{Remark}
\crefname{remark}{Remark}{Remarks}

\crefname{assumption}{Assumption}{Assumptions}

\author{Emma Albertini, Daniel Platt}
\date{\today}
\title{Local uniqueness of the Black String with small circle size}

\usepackage{todonotes}

\begin{document}

\maketitle

\begin{abstract}
    In this article we study uniqueness of the Black String, i.e. the product of 4-dimensional Schwarzschild space with a circle of length $L$. 
    In \cite{Albertini2024}, this was reduced to a non-linear elliptic PDE, and we use this setup to show that for small $L$ the Black String is infinitesimally rigid as a Ricci-flat metric. 
    Using a fixed point theorem, we prove that this implies local uniqueness, i.e. there exist no other Ricci-flat metrics near the Black String, and we give bounds for the size of the neighborhood in which the Black String is unique. 
    We compare this with the toy problem of a scalar field satisfying an elliptic equation that was already solved in \cite{Albertini2024} using different methods. 
    In this case we can use the fixed point theorem method to prove not just a local but a global statement.
\end{abstract}

\tableofcontents

\section{Introduction}
In both astrophysics and quantum gravity, the uniqueness, often referred to as the no-hair property, of stationary black hole solutions plays a crucial role. In four-dimensional spacetime, all regular, stationary, and asymptotically flat solutions of the Einstein–Maxwell equations are completely characterized by just three parameters: mass, angular momentum, and electric charge. These solutions are uniquely described by the Kerr–Newman family of black holes. The collaborative effort of the General Relativity community led to various versions of the proof of the uniqueness for vacuum gravity \cite{Israel_1967, Robinson_1975, Carter_1971, Bunting_1987} and Einstein–Maxwell theory \cite{Mazur_1982, Masood-ul-Alam_1992, Ruback_1988}. It is of great interest to study whether a version of the black hole uniqueness theorem still holds under more general assumptions, such as in the presence of matter fields or in more than four spacetime dimensions, see e.g. \cite{Hollands_2012}.
Results on $4$-dimensional black holes can serve as inspiration in higher dimensions, for instance in cases which exploit the symmetry properties of the “static or axisymmetric” argument in four-dimensions. However, when considering product manifolds, simple connectedness of the event horizon, which is extensively used in the uniqueness proof for Kerr–Newman \cite{Hawking:1971vc, Hawking:1973uf}, does not hold.  
In this context, a noteworthy example exhibiting non-uniqueness are certain static black hole solutions in Kaluza-Klein Theory, i.e. black holes on a product space $M^4 \times S^1$, where ${M}^4$  is a four-dimensional Ricci flat manifold. Black holes are not unique, even when restricting to the case of asymptotics of $\mathcal{M}^4 \times S^1$, where $\mathcal{M}^4$ denotes the Minkowski metric.
In this article we study uniqueness of Black String solutions $\text{Sch}_4 \times S^1_L$, where $\text{Sch}_4$ denotes the $4$-dimensional Schwarzschild space, and $S^1_L=[0,L]/\sim$ denotes a circle of length $L$.
In particular, we will pay attention to the case of small circle sizes $0<L \ll 1$.
This line of investigation was started in \cite{Albertini2024}.
These Black Strings are known to be linearly stable but it is unknown whether they are unique or not.
We do not achieve showing global uniqueness, but prove a weaker rigidity result, namely we show that there exist no nearby Ricci-flat metrics, quantifying exactly the meaning of \emph{nearby}.

The method exploits the elliptic character of Einstein equations and uses elliptic analysis techniques, which lead to an application of Banach's Fixed Point Theorem. 
There are two features of our set up which make our discussion generalizable. 

Firstly, it raises the question about the family of solutions to Einstein equations on higher than five-dimensional manifolds, in particular on product manifolds of the form ${M}^4 \times X$, where ${M}^4$  is a four dimensional manifold and $X$ is any $n$-dimensional compact space. In particular, solutions to the classical Einstein-Maxwell equations in six-dimensional space-time which have the
form of a product of four-dimensional constant curvature space-time with a 2-sphere, namely the Salam-Sezgin Model, may be the natural first example to consider. 
Moreover, since this method relies on a cavity, it can provide a framework to classify black hole solutions in Anti-de-Sitter space.

For $\beta,R >0$ let $D_{\beta, R}$ denote the $2$-dimensional disk with polar coordinates $(\rho,\tau)$ on $D_{\beta, R} \setminus \{0\}$, such that $0<\rho<R$ and $\tau \in S^1_\beta$.
The parameter $\beta$ is the Euclidean time period.

For $\beta,L,R > 0$ consider the space $X_{KK}=D_{\beta,R} \times S^2 \times S^1_L$.
As in \cite{Albertini2024}, we carry out our analysis on a cavity $D_{\beta, R}$ rather than $\R^2$.
As we will take $0<L \ll 1$, the distance of the cavity wall $R$ will appear large compared to $L$. 
Without loss of generality, we take $R=1$.
Going forward, we write $B^2 = D_{\beta,R} \subset \R^2$ for the unit disk in $\R^2$.

We use coordinates $(\tau,x,z) \in B^2 \times S^1_L$, with $\tau,x \in B^2$ and $z \in S^1_L$ and the notation $r=|x|$. We write $\d \Omega^2_{(2)}$ for the metric in sphere coordinates on $S^2$.
In \cite{Albertini2024} it was shown that Ricci flat static SO(3) invariant metrics on $X_{KK}$ with the asymptotics of the product metric of the Schwarzschild metric in Euclidean signature and the circle ($\text{Sch}_4 \times S^1_L$) takes the following form:
\begin{align}
\label{eq:metric-tensor}
g
=
r^2 \d \tau^2
+
e^{2A(r,z)}(\phi(r,z) \d r^2+  \d z^2)
+
\phi(r,z)^{1/2} \d \Omega^2.
\end{align}

Here, $A$ and $\phi$ are real valued functions on $X_L= B^2\times S^1_L$.
(Notice that our convention for the metric differs from \cite[Equation 31]{Albertini2024} in that therein the metric depends on $L$ and the underlying smooth manifold is constant, while we take the point of view that the underlying manifold is $S^1_L=[0,L]/\sim$.)
By \cite[Eqn. 3.31]{Albertini2024}, Ricci-flatness of the metric translates into the following pair of equations for $A$ and $\phi$:
\begin{align} \label{equation:A-K-eqn-without-derivative}
\begin{split}
\phi e^{2 A} \nabla^2 \phi   = 
 4 e^{2 A} \phi^{3/2} + 
\phi_z^2 + 
\frac{\phi_r^2}{\phi}
\\
A_{zz}\phi +\phi_{zz}+\frac{A_r+r A_{rr}}{r}-\frac{5 \phi_z^2}{8\phi}=0,
\end{split}
\end{align}
where $\nabla^2 f=e^{-2A} \left[ \frac{1}{\phi} \left(f_{rr}+\frac{1}{r} f_r \right) + f_{zz} + \frac{\phi_z}{\phi} f_z \right]$.
Although these equations are written in polar coordinates and appear to
have a pole at $r=0$, we always regard the functions involved as
rotationally invariant functions on the Cartesian disk $B^2$. If $f$
and $g$ are radial functions, all polar-coordinate expressions are
understood through the identities
\begin{align*}
    f_{rr}+\frac{1}{r}f_r
    &=\Delta_x f,
    &
    f_r g_r
    &=\nabla_x f\cdot\nabla_x g,
    \\
    f_r^2
    &=|\nabla_x f|^2,
    &
    f_r\partial_r
    &=\nabla_x f\cdot\nabla_x,
\end{align*}
where $\nabla_x$ and $\Delta_x$ are the Cartesian gradient and Laplacian
on $B^2$. For smooth radial functions these Cartesian expressions extend
smoothly across $r=0$, while for Sobolev functions they are understood
in the weak sense. At the cavity wall, $\partial_r$ is the outward
normal derivative.
The boundary condition from \cite[Eqns. 3.36 and 3.37]{Albertini2024} is
\begin{align}
\label{equation:boundary-conditions}
    \begin{split}
    B(A,\phi)=0
    \quad
    \text{ for }
    \quad
    B: C^\infty(X_L, \underline{\R^2}) 
    & \rightarrow
    C^\infty(\partial X_L, \underline{\R^2})
    \\
    (A,\phi)
    &\mapsto
    \left(
    (
    \phi_z|_{\partial X_L},
    \,
    ((2\phi+\phi_r)A_z-(2 \phi_z+\phi_{rz}))|_{\partial X_L}
    \right).
    \end{split}
    \end{align}
The choice $A=0, \phi = \phi_0 := \left(\frac{2\mu}{1-\mu^2r^2}\right)^4$ with $0<\mu<1$ gives the \emph{Homogeneous Black String} metric, which is the Ricci-flat metric given as the product metric of the Schwarzschild metric in Euclidean signature on $S^1 \times \R^3$ and the circle. Note that in order to fix the homogeneous solution, the following boundary conditions are also needed to close the system of PDEs  \cite[Eqns. 3.27, 3.37, and 3.42]{Albertini2024}:
    \begin{align}
    \label{equation:boundary-conditions-homog}
    \begin{split}
        \beta: 
        C^\infty(X_L, \underline{\R^2})
        &\rightarrow \R^2
        \\
        (A,\phi) &\mapsto
        \left(
        \int_{S^1_L} e^{A(1,z)}\,\d z-L,
        \int_{S^1_L} \partial_r\phi(1,z)\,\d z
        -\xi_{\mathrm{homog}}L
        \right),        \text{ where }
        \xi_{homog} = \frac{128 \mu^6}{(1-\mu^2)^5}.
    \end{split}
\end{align}

For small circle size, this metric may be the unique Ricci-flat metric on $X_L$, and the following theorem gives some support to this.
In the following we write $\bar A$ for the zero mode of $A$ and $\bar \phi$ for the zero mode of $\phi$ with respect to the Fourier decomposition in the $z$-direction.

\begin{theorem}
\label{theorem:metric-main-uniqueness}
There exist $c,L_0>0$ such that, for every $0<L<L_0$, the following
holds. Suppose that
$\phi\in L^2_4(X_L)$ and $A\in L^2_3(X_L)$ define a metric of the form
\cref{eq:metric-tensor}, solve
\cref{equation:A-K-eqn-without-derivative}, and satisfy the boundary
conditions
\cref{equation:boundary-conditions,equation:boundary-conditions-homog}.
If
\[
    \|{\phi-\phi_0}_{L^2_4(X_L)}
    \leq
    cL^{1/2},
    \qquad
    \|{A}_{L^2_3(X_L)}
    \leq
    cL^{1/2},
\]
then $\phi=\phi_0$ and $A=0$ everywhere.
\end{theorem}

We briefly explain the proof idea of this theorem.
It has become a standard technique in geometry to construct solutions to elliptic PDE in three steps:
(1) constructing an approximate solution and bounding its error,
(2) bounding the operator norm of the inverse of the linearisation of the equation,
(3) applying a fixed-point theorem to deduce that there exists a genuine solution near the approximate solution.
This method was pioneered in \cite{Taubes1982} and has since been applied to countless other problems in geometry.
Depending on the fixed-point theorem applied, one usually also obtains a \emph{local uniqueness} statement for solutions constructed in this way.
I.e., the method constructs a genuine solution to the elliptic PDE and guarantees that there are no other nearby solutions.

We apply the same method, but instead of beginning with an approximate solution, we begin with the exact known solution to the elliptic equation.
In our case, that is the Black String metric solving the Ricci-flat equation.
The rest of the analysis is then unchanged, and we obtain a \emph{local uniqueness} statement stating that there exist no other nearby Ricci-flat metrics, where we can quantify what \emph{nearby} means.

One difference to examples in the literature is that we do not work with the equations \cref{equation:A-K-eqn-without-derivative} directly, but \emph{with their $z$-derivative}.
While the linearisation of this new system of equations is still a differential operator, it has the interesting consequence that the non-linear terms of the system become integro-differential operators, which requires some adaptation.

Unfortunately, this does not lead to a global uniqueness statement, because we cannot rule out solutions that are far away from the Black String metric.
There is, however, hope that this can be overcome in the future, and we present one toy example in which we can use the method above to prove global uniqueness of solutions to a non-linear elliptic PDE.
The example is a lower-dimensional version of \cite[Section II]{Albertini2024}, in which we are looking for scalar fields on a circle times Minkowski space with a fixed potential $V: \R \rightarrow \R$.
For this problem, we reprove the result from \cite[Section II]{Albertini2024}, which is achieved in \cref{theorem:main-uniqueness}:

\begin{theorem*}
    There exists $L_0>0$ such that the following is true:
    for all $0<L<L_0$, if
    $\phi\in L^2_{3,D}(X_L)$ satisfies
    $\Delta \phi - V'(\phi)=0$,
    then $\phi$ is independent of the $z$-direction.
\end{theorem*}

The article is structured as follows:
in \cref{section:background} we present some facts from analysis that are used later.
In \cref{section:scalar-field-analysis} we prove \cref{theorem:main-uniqueness}, proving the uniqueness for the toy problem of the scalar field.
Last, in \cref{section:metric-analysis} we prove the local uniqueness theorem \cref{theorem:metric-main-uniqueness}.

\section{Background}
\label{section:background}

For every real $k\geq 0$, we write $L^2_k(M)$ for the
$L^2$-based Sobolev space of order $k$, with norm
$\|{f}_{L^2_k(M)}$; in particular, this notation includes fractional
Sobolev spaces. 
See \cite[Chapter~3]{mclean2000strongly} for their definition and the definition of the trace operator
$\operatorname{tr}:
L^2_k(X_L)
\rightarrow
L^2_{k-\frac12}(\partial X_L)$ for $k>1/2$.

\begin{definition}
\label{def:sobolev-spaces}
For $k\geq 0$, define
\[
L^2_{k,\mathrm{mean}-0}(X_L)
:=
\left\{
f\in L^2_k(X_L):
\int_{S^1_L} f(x,z)\,\d z=0
\text{ for almost every }x\in B^2
\right\}.
\]
For $k>\frac12$, define the Dirichlet space
\[
L^2_{k,D}(X_L)
:=
\ker\left(
\operatorname{tr}:
L^2_k(X_L)
\rightarrow
L^2_{k-\frac12}(\partial X_L)
\right).
\]
Finally, set
\[
L^2_{k,D,\mathrm{mean}-0}(X_L)
:=
L^2_{k,D}(X_L)
\cap
L^2_{k,\mathrm{mean}-0}(X_L).
\]
\end{definition}

Even though functions with mean zero need not have zero boundary values, the
Poincaré inequality may be applied in the circle direction and iterated.

\begin{proposition}
\label{proposition:poincare-inequality}
Let $0\leq j\leq k$ be integers and let
$f\in L^2_{k,\mathrm{mean-}0}(X_L)$. Then
\[
    \|{f}_{L^2_j(X_L)}
    \leq
    cL^{k-j}\|{f}_{L^2_k(X_L)},
\]
where $c$ is independent of $L$ and $f$.
\end{proposition}

\begin{proof}
By density, it suffices to consider smooth $f$. 
For every
$\lvert\alpha\rvert\leq j$, the function $D^\alpha f$ has mean zero and
therefore vanishes somewhere on each circle $\{x\}\times S^1_L$.
Applying the Poincaré inequality
\cite[Equation 7.44]{Gilbarg2001} to the circle cut at such a zero, and
then using Fubini's theorem, gives
\begin{align*}
    \|{D^\alpha f}_{L^2(X_L)}^2
    &=
    \int_{B^2}
    \|{D^\alpha f(x,\cdot)}_{L^2(S^1_L)}^2
    \d x
    \\
    &\leq
    cL^2
    \int_{B^2}
    \|{\partial_zD^\alpha f(x,\cdot)}_{L^2(S^1_L)}^2
    \d x
    =
    cL^2\|{\partial_zD^\alpha f}_{L^2(X_L)}^2.
\end{align*}
Every $z$-derivative of $D^\alpha f$ also has mean zero. Iterating this
estimate $k-j$ times and summing over $\lvert\alpha\rvert\leq j$ proves
the claim.
\end{proof}

\begin{remark}
    We use many estimates for the analysis in the article which feature different constants.
    We follow the custom of denoting these constants by $c$, but it is understood this letter can refer to different constants from line to line.
    Important for us is only that a constants exists, not its exact value.
\end{remark}

For $f\in L^2_{k,\mathrm{mean-}0}(X_L)$, let $I(f)$ denote the unique
mean-zero function satisfying $\partial_z I(f)=f$.

\begin{proposition}
\label{proposition:integration-operator}
For every integer $k\geq0$, $I:
    L^2_{k,\mathrm{mean-}0}(X_L)
    \rightarrow
    L^2_{k,\mathrm{mean-}0}(X_L)$
is bounded and satisfies
\[
    \|{I(f)}_{L^2_k(X_L)}
    \leq
    cL\|{f}_{L^2_k(X_L)}.
\]
If $f$ is continuous, then
\[
    \|{I(f)}_{C^0(X_L)}
    \leq
    L\|{f}_{C^0(X_L)}.
\]
\end{proposition}

\begin{proof}
The first estimate holds because $\partial_z$ acts as multiplication by $\frac{2 \pi i n}{L}$ on Fourier modes.
The second estimate holds because on each
circle the continuous mean-zero function $I(f)$ vanishes somewhere, and
the claim follows from the fundamental theorem of calculus.
\end{proof}

Thus, $I$ gains one derivative in the $z$-direction, but it does not gain
a full Sobolev derivative.

We will also use the following Sobolev embeddings with constants explicit
in the circle length.

\begin{proposition}[Sobolev embedding]
\label{proposition:sobolev-embedding}
Fix $L_0>0$. There exists $c>0$ such that, for every $0<L\leq L_0$,
\begin{align}
\label{equation:sobolev-factor-L-minus-1/2}
    \|{f}_{C^j(X_L)}
    &\leq
    cL^{-1/2}\|{f}_{L^2_{j+2}(X_L)},
    && j=0,1,
    \\
\label{equation:sobolev-L4}
    \|{f}_{L^4(X_L)}
    &\leq
    cL^{-1/4}\|{f}_{L^2_1(X_L)}.
\end{align}
\end{proposition}

\begin{proof}
Choose $N\in\mathbb{N}$ such that $L_0\leq NL\leq2L_0$, and lift $f$ to
the $N$-fold cover $X_{NL}\to X_L$. The usual Sobolev embeddings
\cite[Section 4.3]{Taylor2011} have uniform constants on $X_{NL}$ because
$NL\in[L_0,2L_0]$. If $\widetilde f$ denotes the lift, then
\[
    \|{\widetilde f}_{L^2_m(X_{NL})}
    =
    N^{1/2}\|{f}_{L^2_m(X_L)},\qquad
    \|{\widetilde f}_{L^4(X_{NL})}
    =
    N^{1/4}\|{f}_{L^4(X_L)},
\]
while
\[
    \|{\widetilde f}_{C^j(X_{NL})}
    =
    \|{f}_{C^j(X_L)}.
\]
The estimates then follow from $N\leq 2L_0L^{-1}$.
\end{proof}

We use the following elliptic regularity estimate for Dirichlet boundary value problems:

\begin{proposition}
\label{theorem:elliptic-regularity}
Let $P$ be a second-order strongly elliptic scalar differential operator
with smooth coefficients on $X_L$. For $s\in\{0,1\}$ there exists a constant $c(L)>0$, possibly depending on $L$, such that for all $v \in L^2_{s+2}(X_L)$:
\[
    \|{v}_{L^2_{s+2}(X_L)}
    \leq
    c(L)\left(
        \|{Pv}_{L^2_s(X_L)}
        +
        \|{\tr v}_{L^2_{s+\frac32}(\partial X_L)}
        +
        \|{v}_{L^2_{s+1}(X_L)}
    \right).
\]
\end{proposition}

\begin{proof}
The Dirichlet problem is regular by
\cite[Proposition 5.11.10]{Taylor2011}, and the estimate follows from
\cite[Proposition 5.11.2]{Taylor2011}.
\end{proof}

We will need a version of this estimate whose constant is independent of
the circle length.

\begin{corollary}
\label{proposition:linear-estimate-rescaling-argument}
Fix $L_0>0$. For each $0<L\leq L_0$, let $P_L$ be a second-order
strongly elliptic scalar differential operator on $X_L$ whose coefficients are in $C^2(X_L)$ and are independent of $z$. 
Assume that the ellipticity constants and the $C^2$-norms of the coefficients are bounded
independently of $L$. 
Then, for $s\in\{0,1\}$ and $v \in L^2_{s+2}(X_L)$:
\[
    \|{v}_{L^2_{s+2}(X_L)}
    \leq
    c\left(
        \|{P_Lv}_{L^2_s(X_L)}
        +
        \|{\tr v}_{L^2_{s+\frac32}(\partial X_L)}
        +
        \|{v}_{L^2_{s+1}(X_L)}
    \right),
\]
where $c$ is independent of $L$ and $v$.

If, in addition, $v$ has mean zero, then for all sufficiently small $L$,
\[
    \|{v}_{L^2_{s+2}(X_L)}
    \leq
    c\left(
        \|{P_Lv}_{L^2_s(X_L)}
        +
        \|{\tr v}_{L^2_{s+\frac32}(\partial X_L)}
    \right).
\]
\end{corollary}

\begin{proof}
As in the proof of \cref{proposition:sobolev-embedding}, choose $N\in\mathbb{N}$ such that $L_0\leq NL\leq2L_0$, and let
$\pi:X_{NL}\rightarrow X_L$ be the $N$-fold cover. Apply
\cref{theorem:elliptic-regularity} to the lifts of $P_L$ and $v$.
The constants are uniform because $NL\in[L_0,2L_0]$ and the coefficient bounds are uniform. 
The interior norms and fractional boundary norms scale by $N^{1/2}$.
These factors cancel and give the first estimate.

If $v$ has mean zero, \cref{proposition:poincare-inequality} gives
\[
    \|{v}_{L^2_{s+1}(X_L)}
    \leq
    cL\|{v}_{L^2_{s+2}(X_L)}.
\]
For sufficiently small $L$, this term can be absorbed into the
left-hand side.
\end{proof}

We also use the following elementary functional-analytic consequence of
an injectivity estimate.

\begin{proposition}
\label{proposition:bounded-left-inverse}
Let $V$ and $W$ be Hilbert spaces and let $T:V\rightarrow W$ be a
bounded linear map satisfying
$\|{v}_V\leq c\|{Tv}_W$ for all $v\in V$.
Then $T$ has a bounded left inverse $S:W\rightarrow V$ satisfying
$ST=\Id_V$ and $\|{S}\leq c$.
\end{proposition}

This follows from
\cite[Appendix A, Corollary 5.5 and Proposition 6.7]{Taylor2011},
together with the orthogonal decomposition
$W=\Im(T)\oplus\Im(T)^\perp$.

\section{Scalar field analysis}
\label{section:scalar-field-analysis}

In this section we study our toy problem of a scalar with a fixed potential on $\R^2 \times S^1$ (more precisely, a large compact set with boundary therein).
In \cite{Albertini2024}, numerical solutions dependent on the $S^1$-variable were exhibited.
However, heuristics suggest that for small circle size, all solutions must be independent of the $S^1$-variable, and that was already shown to be the case in the reference.
Here, we give an alternative proof for this fact.

We work on $X_L=B^2 \times S^1_L$, where $B^2 \subset \mathbb{R}^2$ denotes the unit disk and $S^1_L=[0,L]/\sim$ is the circle of length $L$.
Often we will denote points in $X_L$ as $(x,z) \in X_L$ with $x \in B^2$ and $z \in S^1_L$ and $r=|x|$.

In the following, we consider $V \in C^3_{\text{loc}}(\R)$ satisfying
\begin{align}
\label{eq:potlbounds}
t V'(t) \ge - c^2 \; , \quad V''(t) \ge - {c}^2 \quad \text{ for all } \; t \in \mathbb{R} \; .
\end{align}

Then we are interested in solutions to the semi-linear elliptic partial differential equation
\begin{align}
        \label{equation:phi-equation}
        \Delta \phi - V'(\phi)=0.
\end{align}

\begin{theorem}
    \label{theorem:main-uniqueness}
    There exists $L_0>0$ such that the following is true:
    for all $0<L<L_0$, if
    $\phi\in L^2_{3,D}(X_L)$ satisfies
    \cref{equation:phi-equation},
    then $\phi$ is independent of the $z$-direction.
\end{theorem}

For a solution $\phi$ as in \cref{theorem:main-uniqueness}, set $\psi:=\partial_z\phi$.
Since $\partial_z$ is tangential to $\partial X_L$ and $\tr\phi=0$, we
have
$\psi\in L^2_{2,D,\mathrm{mean}-0}(X_L)$.
Writing
\[
    \bar\phi(x)
    :=
    \frac{1}{L}\int_{S^1_L}\phi(x,z)\,\d z,
\]
for the constant term in the Fourier decomposition of $\phi$ w.r.t. the $z$-variable, the definition of the integration operator gives
\[
    \phi=\bar\phi+I(\psi).
\]

Differentiating \cref{equation:phi-equation} with respect to $z$, we
obtain
\begin{align}
    \label{equation:psi-equation}
    0
    &=
    \Delta\psi-\psi V''(\phi)
    \notag\\
    &=
    \underbrace{
        \Delta\psi-V''(\bar\phi)\psi
    }_{=:\mathscr{L}(\psi)}
    +
    \underbrace{
        \psi\left(
            V''(\bar\phi)
            -
            V''\bigl(\bar\phi+I(\psi)\bigr)
        \right)
    }_{=:\mathscr{N}(\psi)}.
\end{align}
Here, $\mathscr{L}$ is a second-order elliptic differential operator
whose coefficients are independent of $z$, while $\mathscr{N}$ is a
nonlinear operator involving integration in the $z$-direction.

\subsection{A priori estimate}
\label{subsection:a-priori-estimate}

By using integration by parts on \cref{equation:phi-equation} we can prove certain \emph{a priori} estimates.
That is, assuming a solution $\phi$ to \cref{equation:phi-equation} exists, we prove that $\phi$ must satisfy certain smallness estimates.
This is the content of the following Proposition:

\begin{proposition}
\label{proposition:a-priori-estimates}
There exist $L_0>0$ and $c>0$ such that, for every $0<L<L_0$, the
following holds. If $\phi\in L^2_{3,D}(X_L)$ solves
\cref{equation:phi-equation} and $\psi=\partial_z\phi$, then
\begin{align}
\label{equation:L^2_1-a-priori}
    \|{\phi}_{L^2_1(X_L)}
    &\leq
    cL^{1/2},
    \\
\label{equation:C^0-a-priori}
    \|{\phi}_{C^0(X_L)}
    &\leq
    c,
    \\
\label{equation:L^2_2-a-priori}
    \|{\psi}_{L^2_2(X_L)}
    &\leq
    cL^{1/2},
    \\
\label{equation:I(psi)-a-priori}
    \|{I(\psi)}_{C^0(X_L)}
    &\leq
    cL,
    \\
\label{equation:overline-phi-a-priori}
    \|{\bar\phi}_{C^0(B^2)}
    &\leq
    c.
\end{align}
Here, $c$ is independent of $L$ and of the solution $\phi$.
\end{proposition}

\begin{proof}
Multiplying \cref{equation:phi-equation} by $\phi$ and integrating by
parts gives
\[
    \int_{X_L}|\nabla\phi|^2
    =
    -\int_{X_L}\phi V'(\phi).
\]
The boundary term vanishes because $\tr\phi=0$. Since
$tV'(t)\geq-c^2$, it follows that $\|{\nabla \phi}_{L^2(X_L)}^2
    \leq
    c\operatorname{vol}(X_L)
    =
    cL$.
The Dirichlet Poincaré inequality applied on every slice $B^2\times\{z\}$ and then integrated over $S^1_L$, gives
\[
    \|{\phi}_{L^2(X_L)}
    \leq
    c\|{\nabla_x\phi}_{L^2(X_L)}
    \leq
    c\|{\nabla\phi}_{L^2(X_L)}.
\]
This proves \cref{equation:L^2_1-a-priori}.

We next prove \cref{equation:C^0-a-priori}. Choose $c_V>0$ such that
$tV'(t)\geq-c_V^2$ for every $t\in\R$, and set $u:=\phi^2$.
Using \cref{equation:phi-equation}, we obtain $\Delta u
    =
    2\phi\Delta\phi+2|\nabla\phi|^2
    =
    2\phi V'(\phi)+2|\nabla\phi|^2
    \geq
    -2c_V^2$.
Define
\[
    b(x,z)
    :=
    \frac{c_V^2}{2}(1-r^2),
    \quad
    r=|x|.
\]
Since $\Delta r^2=4$ on $B^2$, we have $\Delta b=-2c_V^2$ and consequently, $\Delta(u-b)\geq 0$.
Moreover, $\partial X_L=\partial B^2\times S^1_L$, so $r=1$ on
$\partial X_L$. Since $\tr\phi=0$, both $u$ and $b$ vanish there.
Thus $\tr(u-b)=0$, and the weak maximum principle \cite[Theorem~8.1]{Gilbarg2001} gives
\[
    u\leq b
    \leq
    \frac{c_V^2}{2}.
\]
Hence $\|{\phi}_{C^0(X_L)}
    \leq
    \frac{c_V}{\sqrt{2}}$,
which proves \cref{equation:C^0-a-priori}.

We now prove \cref{equation:L^2_2-a-priori}. Since
$\psi=\partial_z\phi$, we have $\Delta\psi
    =
    V''(\phi)\psi$.
Furthermore,
$\psi\in L^2_{2,D,\mathrm{mean}-0}(X_L)$. Applying the uniform
Dirichlet estimate \cref{proposition:linear-estimate-rescaling-argument}
to $\Delta$, with $s=0$, gives
\[
    \|{\psi}_{L^2_2(X_L)}
    \leq
    c\|{\Delta\psi}_{L^2(X_L)}
\]
for all sufficiently small $L$. By
\cref{equation:C^0-a-priori}, the image of $\phi$ is contained in a
fixed compact interval. 
Because $V''$ is continuous, we have $\|{V''(\phi)}_{C^0(X_L)}
    \leq
    c$.
Taking both together:
\begin{align*}
    \|{\psi}_{L^2_2(X_L)}
    &\leq
    c\|{V''(\phi)\psi}_{L^2(X_L)}
    \leq
    c\|{\psi}_{L^2(X_L)}
    \leq
    c\|{\phi}_{L^2_1(X_L)}
    \leq
    cL^{1/2},
\end{align*}
where the last inequality follows from
\cref{equation:L^2_1-a-priori}.

By the estimate for the integration operator
\cref{proposition:integration-operator} and the Sobolev embedding
\cref{equation:sobolev-factor-L-minus-1/2}, we have
\begin{align*}
    \|{I(\psi)}_{C^0(X_L)}
    &\leq
    L\|{\psi}_{C^0(X_L)}
    \leq
    cL^{1/2}\|{\psi}_{L^2_2(X_L)}
    \leq
    cL.
\end{align*}
This proves \cref{equation:I(psi)-a-priori}.

Finally, since $\bar\phi=\phi-I(\psi)$,
\[
    \|{\bar\phi}_{C^0(B^2)}
    \leq
    \|{\phi}_{C^0(X_L)}
    +
    \|{I(\psi)}_{C^0(X_L)}
    \leq
    c.
\]
This proves \cref{equation:overline-phi-a-priori}.
\end{proof}

\subsection{The linear operator}
\label{subsection:linear-operator}

In this section we prove an estimate for the inverse of the linear operator $\mathscr{L}$ from \cref{equation:psi-equation}.

\begin{proposition}
\label{proposition:linear-estimate-scalar}
Fix $K>0$. There exist constants $L_0>0$ and $c>0$ such that the
following holds. For every $0<L<L_0$, every $z$-independent function
$\bar\phi$ satisfying $\|{\bar\phi}_{C^0(B^2)}
    \leq
    K$,
and every
$\psi\in L^2_{2,D,\mathrm{mean}-0}(X_L)$, we have
\begin{align}
\label{equation:linear-estimate-scalar}
    \|{\psi}_{L^2_2(X_L)}
    \leq
    c\|{\mathscr L\psi}_{L^2(X_L)}.
\end{align}
Here, $c$ may depend on $K$ and $V$, but it is independent of $L$,
$\bar\phi$, and $\psi$.
\end{proposition}

\begin{proof}
Choose $c_V>0$ such that $V''(t)\geq-c_V^2$ for every  $t\in\R$.
Since $\psi$ has zero trace, integration by parts for Sobolev functions gives
\begin{align*}
    -\int_{X_L}\psi\mathscr L\psi
    &=
    \int_{X_L}|\nabla\psi|^2
    +
    \int_{X_L}V''(\bar\phi)\psi^2
    \geq
    \int_{X_L}(\partial_z\psi)^2
    -
    c_V^2\int_{X_L}\psi^2.
\end{align*}
Since $\psi$ has mean zero on every circle $\{x\}\times S^1_L$, the circle-direction Poincaré estimate from the proof of
\cref{proposition:poincare-inequality} gives
\[
    \int_{X_L}\psi^2
    \leq
    cL^2\int_{X_L}(\partial_z\psi)^2.
\]
Consequently,
\[
    -\int_{X_L}\psi\mathscr L\psi
    \geq
    \left(1-cc_V^2L^2\right)
    \int_{X_L}(\partial_z\psi)^2.
\]
After decreasing $L_0$ if necessary, we therefore have
\[
    -\int_{X_L}\psi\mathscr L\psi
    \geq
    \frac12
    \int_{X_L}(\partial_z\psi)^2
    \geq
    c^{-1}L^{-2}\|{\psi}_{L^2(X_L)}^2,
\]
where the final inequality is another application of the circle-direction Poincaré estimate.
On the other hand, the Cauchy--Schwarz inequality gives
\[
    \left|
        \int_{X_L}\psi\mathscr L\psi
    \right|
    \leq
    \|{\psi}_{L^2(X_L)}
    \|{\mathscr L\psi}_{L^2(X_L)}.
\]
Combining the preceding two estimates gives
\[
    c^{-1}L^{-2}\|{\psi}_{L^2(X_L)}^2
    \leq
    \|{\psi}_{L^2(X_L)}
    \|{\mathscr L\psi}_{L^2(X_L)}.
\]
If $\psi=0$, the desired estimate is immediate. Otherwise, dividing by
$\|{\psi}_{L^2(X_L)}$ gives
\begin{align}
\label{equation:linear-L2-estimate-scalar}
    \|{\psi}_{L^2(X_L)}
    \leq
    cL^2\|{\mathscr L\psi}_{L^2(X_L)}.
\end{align}

We now apply the uniform Dirichlet estimate
\cref{proposition:linear-estimate-rescaling-argument} to the operator
$\Delta$, with $s=0$. Since $\psi$ has zero trace and mean zero, it
gives $\|{\psi}_{L^2_2(X_L)}
    \leq
    c\|{\Delta\psi}_{L^2(X_L)}$.
By the definition of $\mathscr L$,
$\Delta\psi
    =
    \mathscr L\psi
    +
    V''(\bar\phi)\psi$.
Moreover, since $\|{\bar\phi}_{C^0(B^2)}\leq K$ and $V''$ is
continuous,
\[
    \|{V''(\bar\phi)}_{C^0(B^2)}
    \leq
    \sup_{\lvert t\rvert\leq K}|V''(t)|
    \leq
    c.
\]
Using \cref{equation:linear-L2-estimate-scalar}, we conclude that
\begin{align*}
    \|{\psi}_{L^2_2(X_L)}
    &\leq
    c\left(
        \|{\mathscr L\psi}_{L^2(X_L)}
        +
        \|{\psi}_{L^2(X_L)}
    \right)
    \leq
    c\left(1+L^2\right)
    \|{\mathscr L\psi}_{L^2(X_L)}
    \leq
    c\|{\mathscr L\psi}_{L^2(X_L)},
\end{align*}
which proves the claim.
\end{proof}

Although $\mathscr L$ preserves the mean-zero subspace, it will be
convenient to regard it as an operator $\mathscr L:
    L^2_{2,D,\mathrm{mean}-0}(X_L)
    \rightarrow
    L^2(X_L)$.
Later, this will allow us to apply its left inverse to nonlinear terms which need not have mean zero.
The following is a direct consequence of \cref{proposition:bounded-left-inverse}:

\begin{corollary}
\label{corollary:left-inverse-scalar}
For all $M>0$, there exist constants $L_0>0$ and $c>0$ such that the
following holds. For every $0<L<L_0$ and every $z$-independent function
$\bar\phi$ satisfying $\|{\bar\phi}_{C^0(B^2)}
    \leq
    M$,
there exists a bounded left inverse
\[
    S=S_{\bar\phi,L}:
    L^2(X_L)
    \rightarrow
    L^2_{2,D,\mathrm{mean}-0}(X_L)
\]
of $\mathscr L$ satisfying $S \mathscr L
    =
    \Id_{L^2_{2,D,\mathrm{mean}-0}(X_L)}$ and $\|{S}
    \leq
    c$.
Here, $c$ may depend on $M$ and $V$, but it is independent of $L$ and
$\bar\phi$.
\end{corollary}

\subsection{The nonlinear terms}
\label{subsection:nonlinear-terms}

We now prove an estimate for the non-linear terms of \cref{equation:psi-equation}:

\begin{proposition}
\label{proposition:nonlinear-estimate}
Fix $M,K>0$. There exist constants $L_0>0$ and $c>0$ such that the
following holds. Let $0<L<L_0$, let $\bar\phi$ be independent of $z$
and satisfy
$\|{\bar\phi}_{C^0(B^2)}\leq M$, and let
$\psi_1,\psi_2\in L^2_{2,D,\mathrm{mean}-0}(X_L)$ satisfy
$\|{\psi_i}_{L^2_2(X_L)}\leq KL^{1/2}$ for $i=1,2$. Then
\begin{align}
\label{equation:nonlinear-estimate-scalar}
    \|{\mathscr N(\psi_1)-\mathscr N(\psi_2)}_{L^2(X_L)}
    \leq
    cL\|{\psi_1-\psi_2}_{L^2_2(X_L)}.
\end{align}
Here, $c$ may depend on $M$, $K$, and $V$, but it is independent of
$L$, $\bar\phi$, $\psi_1$, and $\psi_2$.
\end{proposition}

\begin{proof}
By \cref{proposition:integration-operator,equation:sobolev-factor-L-minus-1/2}, we have
$\|{I(\psi_i)}_{C^0(X_L)}
\leq cL^{1/2}\|{\psi_i}_{L^2_2(X_L)}
\leq cKL$
for $i=1,2$. 
The same estimates give
$\|{I(\psi_1-\psi_2)}_{C^0(X_L)}
\leq
cL^{1/2}\|{\psi_1-\psi_2}_{L^2_2(X_L)}$.

After decreasing $L_0$ if necessary, we may assume that
$cKL\leq1$. The functions $\bar\phi+I(\psi_i)$ therefore take values
in the fixed compact interval $[-M-1,M+1]$.
By the definition of $\mathscr N$ we have:
\begin{align*}
    \mathscr N(\psi_1)-\mathscr N(\psi_2)
    ={}&
    (\psi_1-\psi_2)
    \left(
        V''(\bar\phi)
        -
        V''\bigl(\bar\phi+I(\psi_1)\bigr)
    \right)
    +
    \psi_2
    \left(
        V''\bigl(\bar\phi+I(\psi_2)\bigr)
        -
        V''\bigl(\bar\phi+I(\psi_1)\bigr)
    \right).
\end{align*}

By the mean value theorem and the fact that $V'''$ is bounded on $[-M-1,M+1]$, the $L^2$-norm of the first term is bounded
by
$c\|{\psi_1-\psi_2}_{L^2(X_L)}
\|{I(\psi_1)}_{C^0(X_L)}$.
The preceding estimates therefore bound it by
$cL\|{\psi_1-\psi_2}_{L^2_2(X_L)}$.

Similarly, the $L^2$-norm of the second term is bounded by
$c\|{\psi_2}_{L^2(X_L)}
\|{I(\psi_1-\psi_2)}_{C^0(X_L)}$.
Since
$\|{\psi_2}_{L^2(X_L)}
\leq
\|{\psi_2}_{L^2_2(X_L)}
\leq
KL^{1/2}$,
this term is also bounded by
$cL\|{\psi_1-\psi_2}_{L^2_2(X_L)}$, which proves the claim.
\end{proof}

\subsection{An application of the contraction mapping theorem}
\label{subsection:scalar-contraction-mapping-theorem}

In this section, we conclude the proof of \cref{theorem:main-uniqueness}.
So far, we have shown two things:
(1) we showed in \cref{subsection:a-priori-estimate} that any solution $\psi$ to the non-linear equation \cref{equation:psi-equation} must satisfy a certain smallness estimate;
(2) taking together the results from \cref{subsection:linear-operator,subsection:nonlinear-terms}, we have that the map $S\mathscr{N}$ is a contraction on a small ball around $0$.
Solutions of \cref{equation:psi-equation} give rise to fixed points.
We already know that $\psi=0$ solves the equation, so taking everything together we will find that $\psi=0$ is the \emph{unique} solution, proving the claim.

\begin{proof}[Proof of \cref{theorem:main-uniqueness}]
Assume that $\phi\in L^2_{3,D}(X_L)$ satisfies
\cref{equation:phi-equation}, and set $\psi:=\partial_z\phi$.
Then $\psi\in L^2_{2,D,\mathrm{mean}-0}(X_L)$ and satisfies
\cref{equation:psi-equation}.
By \cref{proposition:a-priori-estimates}, there are constants
$M,K>0$, independent of $L$ and $\phi$, such that
\[
    \|{\bar\phi}_{C^0(X_L)}
    \leq M,
    \qquad
    \|{\psi}_{L^2_2(X_L)}
    \leq KL^{1/2}.
\]
For $L$ sufficiently small, let $S=S_{\bar\phi,L}$ be the left-inverse
from \cref{corollary:left-inverse-scalar}. Define
\[
    B_{K,L}
    :=
    \left\{
    f\in L^2_{2,D,\mathrm{mean}-0}(X_L):
    \|{f}_{L^2_2(X_L)}\leq KL^{1/2}
    \right\},
    \qquad
    F(f):=-S\mathscr{N}(f).
\]
Then:

\begin{enumerate}
    \item
    The map $F$ maps $B_{K,L}$ into itself for $L$ small enough.

    Indeed, for $f\in B_{K,L}$,
    \[
        \|{F(f)}_{L^2_2(X_L)}
        \leq
        c\|{\mathscr{N}(f)}_{L^2(X_L)}
        \leq
        cL\|{f}_{L^2_2(X_L)}
        \leq
        cKL^{3/2},
    \]
    where we used
    \cref{corollary:left-inverse-scalar,proposition:nonlinear-estimate}
    and the fact that $\mathscr{N}(0)=0$.
    This is at most $KL^{1/2}$ for $L$ small enough.

    \item
    The map $F:B_{K,L}\rightarrow B_{K,L}$ is a contraction for $L$
    small enough.

    For $f_1,f_2\in B_{K,L}$,
    \[
        \|{F(f_1)-F(f_2)}_{L^2_2(X_L)}
        \leq
        c\|{\mathscr{N}(f_1)-\mathscr{N}(f_2)}_{L^2(X_L)}
        \leq
        cL\|{f_1-f_2}_{L^2_2(X_L)}.
    \]
    After making $L$ smaller if necessary, the factor $cL$ is less
    than $1$.

    \item
    The ball $B_{K,L}$ is a closed subset of the Hilbert space
    $L^2_{2,D,\mathrm{mean}-0}(X_L)$ and is therefore complete.
    By the Banach fixed point theorem, it contains a unique fixed point
    of $F$.
    Since $\mathscr{N}(0)=0$, the function $0$ is a fixed point, and
    hence it is the unique fixed point in $B_{K,L}$.

    \item
    By assumption, $0=\mathscr{L}(\psi)+\mathscr{N}(\psi)$.
    Applying $S$ and using $S\mathscr{L}=\Id$ gives $\psi=-S\mathscr{N}(\psi)=F(\psi)$.
    Thus $\psi$ is a fixed point of $F$.
    By the choice of $K$, $\psi\in B_{K,L}$, so the preceding point
    gives $\psi=0$.
    Therefore $\partial_z\phi=0$, and $\phi$ is independent of the
    $z$-direction.
    \qedhere
\end{enumerate}
\end{proof}

\section{Metric analysis}
\label{section:metric-analysis}

In this section we prove
\cref{theorem:metric-main-uniqueness}.
The proof has two steps.
First, we show that a solution satisfying the smallness assumptions in
\cref{theorem:metric-main-uniqueness} is independent of the
$z$-direction.
Second, we show that any $z$-independent solution is equal to $(\phi_0,0)$.

For this purpose, define
\[
    \psi:=\partial_z\phi,
    \qquad
    u:=\partial_z A
\]
so that the from the assumptions of \cref{theorem:metric-main-uniqueness} we have $(\psi,u)
    \in
    L^2_{3,\mathrm{mean}-0}(X_L)
    \oplus
    L^2_{2,\mathrm{mean}-0}(X_L)$.

We write
\[
    \bar\phi(x)
    :=
    \frac{1}{L}\int_{S^1_L}\phi(x,z)\,\d z,
    \qquad
    \bar A(x)
    :=
    \frac{1}{L}\int_{S^1_L}A(x,z)\,\d z
\]
for the zero modes, which gives $\phi=\bar\phi+I(\psi)$, and $A=\bar A+I(u)$.
The hypotheses of \cref{theorem:metric-main-uniqueness} imply
\[
    \|{(\bar\phi-\phi_0,\bar A)}
       _{L^2_4(X_L)\oplus L^2_3(X_L)}
    \leq
    cL^{1/2},
    \qquad
    \|{(\psi,u)}
       _{L^2_3(X_L)\oplus L^2_2(X_L)}
    \leq
    cL^{1/2}.
\]

We now differentiate
\cref{equation:A-K-eqn-without-derivative} in the $z$-direction and
write the resulting system in terms of $(\psi,u)$:
\begin{equation}
\label{equation:K-eqn}
-8 e^{2A} \phi^{3/ 2}  u -6 e^{2A}\phi^{1/2}\psi+ \phi \psi_{zz}+\psi \psi_z+ \psi \frac{\phi_r^2}{\phi^2}-2\frac{\phi_r \psi_r}{\phi}+\frac{ \left(\psi_r+r \psi_{rr}\right)}{r}  =0
\end{equation}
while the A equation reads as
\begin{equation}
\label{equation:A-eqn}
\frac{5 \psi^3}{8\phi^2}+\psi u_z-\frac{5 \psi \psi_z}{4\phi}+\phi u_{zz}+\psi_{zz}+\frac{u_r+r u_{rr}}{ r}=0.
\end{equation}

The linearisation of the system of equations \cref{equation:A-eqn,equation:K-eqn} at $\psi=0$, $u=0$ is the linear differential operator of second order
\begin{align}
    \label{equation:metric-linearisation}
    \begin{split}
    \mathscr{L}:
    L^2_{3,\mathrm{mean}-0}(X_L)
    \oplus
    L^2_{2,\mathrm{mean}-0}(X_L)
    &\longrightarrow
    L^2_1(X_L)
    \oplus
    L^2(X_L)
    \\
    \begin{pmatrix}
        \psi \\ u
    \end{pmatrix}
    &\longmapsto
    \begin{pmatrix}
        e^{2\bar A}\bar\phi\widetilde{\nabla}^2
        -6e^{2\bar A}\bar\phi^{1/2}
        +\dfrac{(\bar\phi')^2}{\bar\phi^2}
        -2\dfrac{\bar\phi'}{\bar\phi}\partial_r
        &
        -8e^{2\bar A}\bar\phi^{3/2}
        \\
        \partial_{zz}
        &
        e^{2\bar A}\bar\phi\widetilde{\nabla}^2
    \end{pmatrix}
    \begin{pmatrix}
        \psi \\ u
    \end{pmatrix},
    \end{split}
\end{align}
where $\tilde \nabla^2= e^{-2\bar A}\left(\frac{1}{\bar \phi}\frac{ \partial_r+ r\partial_{rr}}{  r}+\partial_{zz} \right)$ and $'$ denotes r-derivatives of the zero modes.
In the following, we write $\mathscr{N}$ for the nonlinear part of the system of equations \cref{equation:A-eqn,equation:K-eqn}.

The linearisation and higher order terms respectively of the boundary operator $B$ from \cref{equation:boundary-conditions} at $(A,\phi)=(0,\phi_0)$ are
\begin{align*}
    B_{\mathscr L}:
    L^2_{3,\mathrm{mean}-0}(X_L)
    \oplus
    L^2_{2,\mathrm{mean}-0}(X_L)
    &\longrightarrow
    L^2_{5/2}(\partial X_L)
    \oplus
    L^2_{3/2}(\partial X_L)
    \\
    (\psi,u)
    &\longmapsto
    \left(
        \tr\psi,
        \tr\left(
            (2\bar\phi+\partial_r\bar\phi)u
            -(2\psi+\partial_r\psi)
        \right)
    \right),
    \\
    B_{\mathscr N} := B-B_{\mathscr{L}}:
    L^2_{3,\mathrm{mean}-0}(X_L)
    \oplus
    L^2_{2,\mathrm{mean}-0}(X_L)
    &\longrightarrow
    L^2_{5/2}(\partial X_L)
    \oplus
    L^2_{3/2}(\partial X_L)
    \\
    (\psi,u)
    &\longmapsto
    \left(
        0,
        \tr\left(
            (2I(\psi)+\partial_r I(\psi))u
        \right)
    \right).
\end{align*}

\subsection{The linear terms}

We will prove:

\begin{proposition}
\label{proposition:metric-linear}
There exist constants $c_0,L_0,C>0$ such that the following is true.
Let $0<L<L_0$, and let
$\bar\phi\in L^2_4(X_L)$ and $\bar A\in L^2_3(X_L)$ be independent of
the $z$-variable and satisfy
\[
    \|{\bar A}_{L^2_3(X_L)}
    \leq
    c_0L^{1/2},
    \quad
    \|{\bar\phi-\phi_0}_{L^2_4(X_L)}
    \leq
    c_0L^{1/2}.
\]
Then
\[
    (\mathscr L,B_{\mathscr L}):
    L^2_{3,\mathrm{mean}-0}(X_L)
    \oplus
    L^2_{2,\mathrm{mean}-0}(X_L)
    \longrightarrow
    L^2_1(X_L)
    \oplus
    L^2(X_L)
    \oplus
    L^2_{5/2}(\partial X_L)
    \oplus
    L^2_{3/2}(\partial X_L)
\]
satisfies
\[
    \|{(\psi,u)}_{L^2_3(X_L)\oplus L^2_2(X_L)}
    \leq
    C\left(
        \|{\mathscr L(\psi,u)}_{L^2_1(X_L)\oplus L^2(X_L)}
        +
        \|{B_{\mathscr L}(\psi,u)}_{
            L^2_{5/2}(\partial X_L)
            \oplus
            L^2_{3/2}(\partial X_L)}
    \right)
\]
for every
$(\psi,u)\in
L^2_{3,\mathrm{mean}-0}(X_L)
\oplus
L^2_{2,\mathrm{mean}-0}(X_L)$.
The constant $C$ is independent of $L$, $\bar A$, and $\bar\phi$.
\end{proposition}

\begin{proof}
We see, as in \cref{proposition:a-priori-estimates}, that $\|{\bar A}_{L^2_3(B^2)}
    +
    \|{\bar\phi-\phi_0}_{L^2_4(B^2)}
    \leq
    c_0$.
By the Sobolev embedding theorem on $B^2$, if $c_0$ is sufficiently
small, there exist constants $m,b,M>0$, independent of $L$, such that
$\bar\phi\geq m$,
$\|{\bar\phi}_{C^2(B^2)}+\|{\bar A}_{C^1(B^2)}\leq M$, and
\begin{align}
\label{eqn:linearised-proof-boundary-bound}
    \left.
    (2\bar\phi+\partial_r\bar\phi)
    \right|_{\partial X_L}
    \geq b.
\end{align}

\textbf{Step 1: Interior estimate.}
Denote the $i$-th component of $\mathscr L$ by $\mathscr L_i$.
Applying \cref{proposition:linear-estimate-rescaling-argument} to
$e^{2\bar A}\bar\phi\widetilde{\nabla}^2\psi$, and then using the
definition of $\mathscr L_1$, gives
\begin{align}
\label{equation:linearisation-proof1}
    \|{\psi}_{L^2_3(X_L)}
    &\leq
    C\left(
        \|{
            e^{2\bar A}\bar\phi
            \widetilde{\nabla}^2\psi
        }_{L^2_1(X_L)}
        +
        \|{\tr\psi}_{L^2_{5/2}(\partial X_L)}
    \right)
    \nonumber\\
    &\leq
    C\left(
        \|{\mathscr L_1(\psi,u)}_{L^2_1(X_L)}
        +
        \|{B_{\mathscr L}(\psi,u)}_{
            L^2_{5/2}(\partial X_L)
            \oplus
            L^2_{3/2}(\partial X_L)}
        +
        \|{\psi}_{L^2_2(X_L)}
        +
        \|{u}_{L^2_1(X_L)}
    \right).
\end{align}
Here the last two terms contain all the lower-order terms in
$\mathscr L_1$; their coefficients are uniformly bounded in the
required $C^1$-norms.

Solving the second component of
$B_{\mathscr L}$ for $\tr u$ gives
\[
    \tr u
    =
    \frac{
        \tr\left(
            (2\bar\phi+\partial_r\bar\phi)u
            -(2\psi+\partial_r\psi)
        \right)
        +
        2\tr\psi
        +
        \tr(\partial_r\psi)
    }{
        \left.
        (2\bar\phi+\partial_r\bar\phi)
        \right|_{\partial X_L}
    }.
\]
Consequently,
\begin{align*}
    \|{\tr u}_{L^2_{3/2}(\partial X_L)}
    &\leq
    C\left(
        \|{
            \tr\left(
                (2\bar\phi+\partial_r\bar\phi)u
                -(2\psi+\partial_r\psi)
            \right)
        }_{L^2_{3/2}(\partial X_L)}
        +
        \|{\tr\psi}_{L^2_{3/2}(\partial X_L)}
        +
        \|{\tr(\partial_r\psi)}_{L^2_{3/2}(\partial X_L)}
    \right)
    \\
    &\leq
    C\left(
        \|{B_{\mathscr L}(\psi,u)}_{
            L^2_{5/2}(\partial X_L)
            \oplus
            L^2_{3/2}(\partial X_L)}
        +
        \|{\psi}_{L^2_3(X_L)}
    \right).
\end{align*}
In the first step we used \cref{eqn:linearised-proof-boundary-bound}. In the second step we used the bound \cite[Theorem~3.37]{mclean2000strongly} for the trace operator.
That this bound is uniform in $L$ follows from the same covering argument as in
the proof of
\cref{proposition:linear-estimate-rescaling-argument}.
Applying \cref{proposition:linear-estimate-rescaling-argument} to
$e^{2\bar A}\bar\phi\widetilde{\nabla}^2u$, and using the definition of
$\mathscr L_2$, now gives
\begin{align}
\label{equation:linearisation-proof2}
    \|{u}_{L^2_2(X_L)}
    &\leq
    C\left(
        \|{
            e^{2\bar A}\bar\phi
            \widetilde{\nabla}^2u
        }_{L^2(X_L)}
        +
        \|{\tr u}_{L^2_{3/2}(\partial X_L)}
    \right)
    \nonumber\\
    &\leq
    C\left(
        \|{\mathscr L_2(\psi,u)}_{L^2(X_L)}
        +
        \|{\psi}_{L^2_2(X_L)}
        +
        \|{B_{\mathscr L}(\psi,u)}_{
            L^2_{5/2}(\partial X_L)
            \oplus
            L^2_{3/2}(\partial X_L)}
        +
        \|{\psi}_{L^2_3(X_L)}
    \right).
\end{align}

Substituting \cref{equation:linearisation-proof1} into
\cref{equation:linearisation-proof2} yields
\[
    \|{(\psi,u)}_{L^2_3(X_L)\oplus L^2_2(X_L)}
    \leq
    C\left(
        \|{\mathscr L(\psi,u)}_{L^2_1(X_L)\oplus L^2(X_L)}
        +
        \|{B_{\mathscr L}(\psi,u)}_{
            L^2_{5/2}(\partial X_L)
            \oplus
            L^2_{3/2}(\partial X_L)}
        +
        \|{(\psi,u)}_{L^2_2(X_L)\oplus L^2_1(X_L)}
    \right).
\]

\textbf{Step 2: Absorption of the lower-order norms.}
Since $\psi$ and $u$ have mean zero,
\cref{proposition:poincare-inequality} gives
\begin{align*}
    \|{(\psi,u)}_{L^2_3(X_L)\oplus L^2_2(X_L)}
    &\leq
    C\left(
        \|{\mathscr L(\psi,u)}_{L^2_1(X_L)\oplus L^2(X_L)}
        +
        \|{B_{\mathscr L}(\psi,u)}_{
            L^2_{5/2}(\partial X_L)
            \oplus
            L^2_{3/2}(\partial X_L)}
    \right)
    \\
    &\quad
    +
    CL
    \|{(\psi,u)}_{L^2_3(X_L)\oplus L^2_2(X_L)}.
\end{align*}
For sufficiently small $L$, the last term can be absorbed into the
left-hand side. This proves the claim.
\end{proof}

Just as in \cref{section:scalar-field-analysis} we obtain a bounded left inverse:

\begin{corollary}
\label{corollary:left-inverse-metric}
Under the assumptions of
\cref{proposition:metric-linear}, there exists a bounded linear map
\[
    T_{\bar\phi,\bar A,L}:
    L^2_1(X_L)
    \oplus
    L^2(X_L)
    \oplus
    L^2_{5/2}(\partial X_L)
    \oplus
    L^2_{3/2}(\partial X_L)
    \longrightarrow
    L^2_{3,\mathrm{mean}-0}(X_L)
    \oplus
    L^2_{2,\mathrm{mean}-0}(X_L)
\]
such that
\[
    T_{\bar\phi,\bar A,L}
    (\mathscr L,B_{\mathscr L})
    =
    \Id_{
        L^2_{3,\mathrm{mean}-0}(X_L)
        \oplus
        L^2_{2,\mathrm{mean}-0}(X_L)}
\]
and
\[
    \|{T_{\bar\phi,\bar A,L}}
    \leq c,
\]
where $c$ is independent of $L$, $\bar\phi$, and $\bar A$.
\end{corollary}

This follows immediately from
\cref{proposition:metric-linear,proposition:bounded-left-inverse}.

\subsection{The nonlinear terms}

For $(\psi,u)
    \in
    L^2_{3,\mathrm{mean}-0}(X_L)
    \oplus
    L^2_{2,\mathrm{mean}-0}(X_L)$ set $\phi=\bar\phi+I(\psi),
    A=\bar A+I(u)$.
The nonlinear interior
terms are
\begin{align*}
\mathscr{N}_1(\psi,u)
={}&
I(\psi)\psi_{zz}
+
\psi\psi_z
-
8\left(
e^{2A}\phi^{3/2}
-
e^{2\bar A}\bar\phi^{3/2}
\right)u
\\
&-
6\left(
e^{2A}\phi^{1/2}
-
e^{2\bar A}\bar\phi^{1/2}
\right)\psi
\\
&-
2\left(
(\log\phi)_r-(\log\bar\phi)_r
\right)\psi_r
+
\left(
(\log\phi)_r^2-(\log\bar\phi)_r^2
\right)\psi,
\\
\mathscr{N}_2(\psi,u)
={}&
I(\psi)u_{zz}
+
\psi u_z
-
\frac{5\psi\psi_z}{4\phi}
+
\frac{5\psi^3}{8\phi^2}
\end{align*}
where $\mathscr{N}=(\mathscr{N}_1,\mathscr{N}_2)$ and we used the convention from \cref{equation:A-K-eqn-without-derivative} of interpreting radial derivatives in Cartesian coordinates.
The nonlinear boundary term is
\[
B_{\mathscr{N}}(\psi,u)
=
\left(
0,
\left(
2\tr I(\psi)+\tr (\partial_r I(\psi))
\right)\tr u
\right).
\]

\begin{proposition}
\label{proposition:metric-nonlinear-estimate}
Fix $K>0$, and let $c_0>0$ be sufficiently small. There exist
constants $L_0,c>0$ such that the following holds.
Let $0<L<L_0$, and let
$\bar\phi\in L^2_4(X_L)$ and $\bar A\in L^2_3(X_L)$ be independent of
$z$ and satisfy
\[
    \|{\bar\phi-\phi_0}_{L^2_4(X_L)}
    \leq
    c_0L^{1/2},
    \qquad
    \|{\bar A}_{L^2_3(X_L)}
    \leq
    c_0L^{1/2}.
\]
Suppose that
\[
    (\psi_i,u_i)
    \in
    L^2_{3,\mathrm{mean}-0}(X_L)
    \oplus
    L^2_{2,\mathrm{mean}-0}(X_L),
    \qquad i\in\{1,2\},
\]
and
\[
    \|{(\psi_i,u_i)}_
    {L^2_3(X_L)\oplus L^2_2(X_L)}
    \leq
    KL^{1/2}.
\]
Then
\begin{align}
\label{equation:metric-nonlinear-estimate}
&\quad
\|{
(\mathscr{N},B_{\mathscr{N}})(\psi_1,u_1)
-
(\mathscr{N},B_{\mathscr{N}})(\psi_2,u_2)
}_{
L^2_1(X_L)\oplus L^2(X_L)
\oplus L^2_{5/2}(\partial X_L)
\oplus L^2_{3/2}(\partial X_L)
}
\nonumber\\
&\leq
cL
\|{
(\psi_1,u_1)-(\psi_2,u_2)
}_{
L^2_3(X_L)\oplus L^2_2(X_L)
}.
\end{align}
Here, $c$ may depend on $c_0$ and $K$, but it is independent of $L$.
\end{proposition}

\begin{proof}
Write $\phi_i=\bar\phi+I(\psi_i)$, $A_i=\bar A+I(u_i)$ for $i\in\{1,2\}$
and set
$d
    :=
    \|{
    (\psi_1,u_1)-(\psi_2,u_2)
    }_{
    L^2_3(X_L)\oplus L^2_2(X_L)
    }$.

By \cref{equation:sobolev-factor-L-minus-1/2} and the assumptions on
the zero modes,
\[
    \|{\bar\phi-\phi_0}_{C^1(X_L)}
    +
    \|{\bar A}_{C^1(X_L)}
    \leq
    cL^{-1/2}
    \left(
        \|{\bar\phi-\phi_0}_{L^2_4(X_L)}
        +
        \|{\bar A}_{L^2_3(X_L)}
    \right)
    \leq
    cc_0.
\]
After making $c_0$ sufficiently small, $\bar\phi$ is therefore bounded
away from zero.
By
\cref{proposition:integration-operator,proposition:sobolev-embedding},
\begin{align}
\label{equation:metric-nonlinear-working-bounds}
    \|{I(\psi_i)}_{L^2_3(X_L)}
    +
    \|{I(u_i)}_{L^2_2(X_L)}
    &\leq
    cKL^{3/2},
    \\
    \|{I(\psi_i)}_{C^1(X_L)}
    +
    \|{I(u_i)}_{C^0(X_L)}
    &\leq
    cKL,
    \nonumber
\end{align}
while \cref{proposition:integration-operator} and \cref{equation:sobolev-factor-L-minus-1/2} give
\begin{align}
\label{eqn:I-difference-bound}
    \|{I(\psi_1-\psi_2)}_{L^2_3(X_L)}
    +
    \|{I(u_1-u_2)}_{L^2_2(X_L)}
    \leq
    cLd.
\end{align}
After making $L_0$ smaller if necessary, the functions
$\bar\phi,\phi_1,\phi_2$ therefore take values in a fixed compact
subset of $(0,\infty)$.
We also obtain from
\cref{proposition:poincare-inequality} that
\begin{align}
\label{equation:metric-nonlinear-poincare-bounds}
    \|{\psi_i}_{L^2_2(X_L)}
    +
    \|{u_i}_{L^2_1(X_L)}
    &\leq
    cKL^{3/2},
    \\
    \|{\psi_1-\psi_2}_{L^2_2(X_L)}
    +
    \|{u_1-u_2}_{L^2_1(X_L)}
    &\leq
    cLd.
    \nonumber
\end{align}

We will frequently use the following consequence of the Sobolev embedding theorem
\cref{proposition:sobolev-embedding}:
\begin{align}
\label{equation:metric-gradient-product-estimate}
    \|{\nabla(fg)}_{L^2(X_L)}
    &\leq
    \|{f}_{C^0(X_L)}
    \|{\nabla g}_{L^2(X_L)}
    +
    \|{\nabla f}_{L^4(X_L)}
    \|{g}_{L^4(X_L)}
    \leq
    cL^{-1/2}
    \|{f}_{L^2_2(X_L)}
    \|{g}_{L^2_1(X_L)}.
\end{align}
The same argument applied to the zeroth-order part and second derivatives gives
\begin{equation}
\label{equation:metric-H2-product-estimate}
\begin{split}
    \|{fg}_{L^2_1(X_L)}
    \leq
    cL^{-1/2}
    \|{f}_{L^2_2(X_L)}
    \|{g}_{L^2_1(X_L)}
    \\
    \|{fg}_{L^2_2(X_L)}
    \leq
    cL^{-1/2}
    \|{f}_{L^2_2(X_L)}
    \|{g}_{L^2_2(X_L)}.
\end{split}
\end{equation}
The estimates used below have constants independent of $L$, which can be seen as in the proof of \cref{proposition:sobolev-embedding}.
For $q\in\{\frac12,\frac32\}$, we have
\begin{align}
&\|{
e^{2A_1}\phi_1^q
-
e^{2A_2}\phi_2^q
}_{L^2_2(X_L)}
\nonumber\\
&\leq
c\Bigg(
    \|{A_1-A_2}_{L^2_2(X_L)}
    +
    \|{\phi_1-\phi_2}_{L^2_2(X_L)}
\nonumber\\
&\qquad
    +
    \left(
        \|{A_1-A_2}_{L^\infty(X_L)}
        +
        \|{\phi_1-\phi_2}_{L^\infty(X_L)}
    \right)
\nonumber\\
&\qquad\quad\cdot
    \left(
        1+
        \sum_{i=1}^2
        \left(
            \|{A_i}_{L^\infty(X_L)}
            +
            \|{\phi_i}_{L^\infty(X_L)}
        \right)
    \right)^2
\nonumber\\
&\qquad\quad\cdot
    \left(
        L^{1/2}
        +
        \sum_{i=1}^2
        \left(
            \|{A_i}_{L^2_2(X_L)}
            +
            \|{\phi_i}_{L^2_2(X_L)}
        \right)
    \right)
\Bigg)
\nonumber\\
&\leq
c\left(
    Ld
    +
    L^{1/2}d(1+c)^2L^{1/2}
\right)
\leq
cLd.
\label{equation:metric-coefficient-composition}
\end{align}
In the first step we split
$e^{2A_1}\phi_1^q-e^{2A_2}\phi_2^q
=(e^{2A_1}-e^{2A_2})\phi_1^q
+e^{2A_2}(\phi_1^q-\phi_2^q)$
and used the local Lipschitz estimate for composition operators
\cite[Theorem~5.5.2.1]{Runst2011}, with
$F^2_{2,2}=L^2_2$ and $\mu=3$, together with the multiplication estimate
\cite[Theorem~4.6.4.2]{Runst2011}.
In the second step we used the assumptions together with
\cref{proposition:integration-operator} and \cref{equation:sobolev-factor-L-minus-1/2}.

Analogously, replacing $(A_2,\phi_2)$ by $(\bar A,\bar\phi)$ and using
$A_i-\bar A=I(u_i)$ and $\phi_i-\bar\phi=I(\psi_i)$, we obtain
\begin{equation}
\label{equation:metric-coefficient-background-composition}
    \|{
    e^{2A_i}\phi_i^q
    -
    e^{2\bar A}\bar\phi^q
    }_{L^2_2(X_L)}
    \leq
    cL
    \left(
        \|{\psi_i}_{L^2_3(X_L)}
        +
        \|{u_i}_{L^2_2(X_L)}
    \right).
\end{equation}

The same argument gives
\begin{align}
\label{equation:metric-log-composition}
    \|{\log\phi_i-\log\bar\phi}_{L^2_3(X_L)}
    &\leq
    cL\|{\psi_i}_{L^2_3(X_L)},
    \\
    \|{\log\phi_1-\log\phi_2}_{L^2_3(X_L)}
    &\leq
    cL\|{\psi_1-\psi_2}_{L^2_3(X_L)}.
    \nonumber
\end{align}
Finally, $\bar\phi$ is independent of $z$ and belongs to a fixed bounded
subset of $L^2_4(B^2)$ whose elements are uniformly positive. Hence
$\|{\log\bar\phi}_{L^2_3(X_L)}\leq cL^{1/2}$. Combining this with the
first estimate in \cref{equation:metric-log-composition} gives
\begin{equation}
\label{equation:metric-log-background-bound}
    \|{\log\phi_i}_{L^2_3(X_L)}
    +
    \|{\log\bar\phi}_{L^2_3(X_L)}
    \leq
    cL^{1/2}.
\end{equation}

\textbf{Step 1: Interior terms.}

We begin with the terms in $\mathscr{N}_1$. 
Using
\cref{equation:metric-H2-product-estimate} in the first step and
\cref{equation:metric-nonlinear-working-bounds,eqn:I-difference-bound}
in the second step, we get
\begin{align*}
&\quad
\|{
I(\psi_1)\psi_{1,zz}
-
I(\psi_2)\psi_{2,zz}
}_{L^2_1(X_L)}
\\
&\leq
cL^{-1/2}
\left(
\|{I(\psi_1)}_{L^2_2(X_L)}
\|{\psi_1-\psi_2}_{L^2_3(X_L)}
+
\|{I(\psi_1-\psi_2)}_{L^2_2(X_L)}
\|{\psi_2}_{L^2_3(X_L)}
\right)
\leq
cLd.
\end{align*}
Similarly, using
\cref{equation:metric-H2-product-estimate} in the first step and
\cref{equation:metric-nonlinear-poincare-bounds} in the second step,
\begin{align}
\label{equation:metric-psi-psiz-difference}
\|{\psi_1\psi_{1,z}-\psi_2\psi_{2,z}}_{L^2_1(X_L)}
&\leq
cL^{-1/2}
\left(
\|{\psi_1}_{L^2_2(X_L)}
\|{\psi_1-\psi_2}_{L^2_3(X_L)}
+
\|{\psi_1-\psi_2}_{L^2_2(X_L)}
\|{\psi_2}_{L^2_3(X_L)}
\right)
\leq
cLd.
\end{align}

We next estimate the term containing the nonlinear coefficient
$e^{2A}\phi^{3/2}-e^{2\bar A}\bar\phi^{3/2}$:
\begin{align*}
&\quad
\|{
\left(
e^{2A_1}\phi_1^{3/2}
-
e^{2\bar A}\bar\phi^{3/2}
\right)u_1
-
\left(
e^{2A_2}\phi_2^{3/2}
-
e^{2\bar A}\bar\phi^{3/2}
\right)u_2
}_{L^2_1(X_L)}
\\
&\leq
cL^{-1/2}
\left[
\|{
e^{2A_1}\phi_1^{3/2}
-
e^{2\bar A}\bar\phi^{3/2}
}_{L^2_2(X_L)}
\|{u_1-u_2}_{L^2_1(X_L)}
\right.
\\
&\hspace{35mm}\left.
+
\|{
e^{2A_1}\phi_1^{3/2}
-
e^{2A_2}\phi_2^{3/2}
}_{L^2_2(X_L)}
\|{u_2}_{L^2_1(X_L)}
\right]
\\
&\leq
cLd.
\end{align*}
In the first step we used
\cref{equation:metric-H2-product-estimate}, \cref{equation:metric-coefficient-composition}, \cref{equation:metric-coefficient-background-composition};
in the second step we used
\cref{equation:metric-nonlinear-poincare-bounds}.
The term containing
$e^{2A}\phi^{1/2}-e^{2\bar A}\bar\phi^{1/2}$ and final factor $\psi$
is estimated in the same way.

We next consider the first term containing radial derivatives. We use
its Cartesian interpretation and display only the differentiated part
of the $L^2_1$-norm, since the $L^2$-part is easier:
\begin{align*}
&\quad
\|{
\nabla\left[
\left(
\nabla_x\log\phi_1-\nabla_x\log\bar\phi
\right)\cdot\nabla_x\psi_1
-
\left(
\nabla_x\log\phi_2-\nabla_x\log\bar\phi
\right)\cdot\nabla_x\psi_2
\right]
}_{L^2(X_L)}
\\
&=
\|{
\nabla\left[
\underbrace{
\nabla_x(\log\phi_1-\log\bar\phi)
}_{=f_1}
\cdot
\underbrace{
\nabla_x(\psi_1-\psi_2)
}_{=g_1}
+
\underbrace{
\nabla_x(\log\phi_1-\log\phi_2)
}_{=f_2}
\cdot
\underbrace{
\nabla_x\psi_2
}_{=g_2}
\right]
}_{L^2(X_L)}
\\
&\leq
cL^{-1/2}
\left(
\|{\log\phi_1-\log\bar\phi}_{L^2_3(X_L)}
\|{\psi_1-\psi_2}_{L^2_2(X_L)}
\right.
\\
&\hspace{35mm}\left.
+
\|{\log\phi_1-\log\phi_2}_{L^2_3(X_L)}
\|{\psi_2}_{L^2_2(X_L)}
\right)
\\
&\leq
cLd.
\end{align*}
In the first inequality we used
\cref{equation:metric-gradient-product-estimate}; in the last
inequality we used
\cref{equation:metric-log-composition}, \cref{equation:metric-nonlinear-poincare-bounds} and the assumptions of the
proposition.

We estimate the remaining radial term:
\begin{align*}
&\quad
\left(
|\nabla_x\log\phi_1|^2
-
|\nabla_x\log\bar\phi|^2
\right)\psi_1
-
\left(
|\nabla_x\log\phi_2|^2
-
|\nabla_x\log\bar\phi|^2
\right)\psi_2
\\
&=
\left(
|\nabla_x\log\phi_1|^2
-
|\nabla_x\log\bar\phi|^2
\right)
(\psi_1-\psi_2)
+
\left(
|\nabla_x\log\phi_1|^2
-
|\nabla_x\log\phi_2|^2
\right)\psi_2
\\
&=
\nabla_x(\log\phi_1-\log\bar\phi)
\cdot
\nabla_x(\log\phi_1+\log\bar\phi)
(\psi_1-\psi_2)
+
\nabla_x(\log\phi_1-\log\phi_2)
\cdot
\nabla_x(\log\phi_1+\log\phi_2)
\psi_2.
\end{align*}
Applying \cref{equation:metric-H2-product-estimate} twice, both terms
on the right-hand side have $L^2_1$-norm bounded by $cLd$.
Here, the logarithmic factors are controlled by
\cref{equation:metric-log-composition} and
\cref{equation:metric-log-background-bound}, while the remaining
factors are controlled by
\cref{equation:metric-nonlinear-poincare-bounds}.

The terms in $\mathscr{N}_2$ only need to be estimated in $L^2$. Using
\cref{equation:metric-nonlinear-working-bounds,eqn:I-difference-bound},
\begin{align*}
&\quad
\|{
I(\psi_1)u_{1,zz}
-
I(\psi_2)u_{2,zz}
}_{L^2(X_L)}
\\
&\leq
\|{I(\psi_1)}_{C^0(X_L)}
\|{u_1-u_2}_{L^2_2(X_L)}
+
\|{I(\psi_1-\psi_2)}_{C^0(X_L)}
\|{u_2}_{L^2_2(X_L)}
\leq
cLd.
\end{align*}
Moreover,
\cref{proposition:poincare-inequality,proposition:sobolev-embedding}
give
\begin{equation}
\label{equation:metric-psi-C0}
    \|{\psi}_{C^0(X_L)}
    \leq
    cL^{1/2}\|{\psi}_{L^2_3(X_L)}
\end{equation}
for every mean-zero $\psi$. Hence
\begin{align*}
\|{\psi_1u_{1,z}-\psi_2u_{2,z}}_{L^2(X_L)}
&\leq
\|{\psi_1}_{C^0(X_L)}
\|{u_1-u_2}_{L^2_1(X_L)}
+
\|{\psi_1-\psi_2}_{C^0(X_L)}
\|{u_2}_{L^2_1(X_L)}
\leq
cLd.
\end{align*}

In the last step we used
\cref{equation:metric-psi-C0}, \cref{equation:metric-nonlinear-poincare-bounds} and the assumptions of the
proposition.

Since $\phi_1$ and $\phi_2$ take values in a fixed compact subset of
$(0,\infty)$, the mean-value theorem and
\cref{eqn:I-difference-bound} give, for $m\in\{1,2\}$,
\begin{equation}\label{equation:metric-reciprocal-difference}
    \|{\phi_1^{-m}-\phi_2^{-m}}_{C^0(X_L)}
    \leq
    c\|{I(\psi_1-\psi_2)}_{C^0(X_L)}
    \leq
    cL^{1/2}d.
\end{equation}
Using
\cref{equation:metric-psi-psiz-difference},\cref{equation:metric-psi-C0},\cref{equation:metric-reciprocal-difference}, \cref{equation:metric-nonlinear-working-bounds},
we obtain
\begin{align*}
\|{
\frac{\psi_1\psi_{1,z}}{\phi_1}
-
\frac{\psi_2\psi_{2,z}}{\phi_2}
}_{L^2(X_L)}
&\leq
c\|{\psi_1\psi_{1,z}-\psi_2\psi_{2,z}}_{L^2(X_L)}
\\
&\quad+
c\|{\psi_2}_{C^0(X_L)}
\|{\psi_{2,z}}_{L^2(X_L)}
\|{\phi_1^{-1}-\phi_2^{-1}}_{C^0(X_L)}
\\
&\leq
cLd.
\end{align*}
Likewise, using
\cref{equation:metric-psi-C0}, \cref{equation:metric-reciprocal-difference}, \cref{equation:metric-nonlinear-working-bounds},
\begin{align*}
\|{
\frac{\psi_1^3}{\phi_1^2}
-
\frac{\psi_2^3}{\phi_2^2}
}_{L^2(X_L)}
&\leq
c\left(
\|{\psi_1}_{C^0(X_L)}^2
+
\|{\psi_2}_{C^0(X_L)}^2
\right)
\|{\psi_1-\psi_2}_{L^2(X_L)}
\\
&\quad+
c\|{\psi_2}_{C^0(X_L)}^2
\|{\psi_2}_{L^2(X_L)}
\|{\phi_1^{-2}-\phi_2^{-2}}_{C^0(X_L)}
\leq
cL^2d.
\end{align*}
Consequently,
\[
\|{
\mathscr{N}(\psi_1,u_1)
-
\mathscr{N}(\psi_2,u_2)
}_{L^2_1(X_L)\oplus L^2(X_L)}
\leq
cLd.
\]

\textbf{Step 2: Boundary terms.}

By the trace theorem
\cite[Theorem~3.37]{mclean2000strongly} and
\cref{proposition:integration-operator}, we have
for every $\eta\in L^2_{3,\mathrm{mean}-0}(X_L)$:
\begin{align}
\label{equation:metric-normal-trace}
    \|{\tr I(\eta)}_{L^2_{5/2}(\partial X_L)}
    +
    \|{\tr(\partial_r I(\eta))}_{L^2_{3/2}(\partial X_L)}
    &\leq
    c\|{I(\eta)}_{L^2_3(X_L)}
    \leq
    cL\|{\eta}_{L^2_3(X_L)}.
\end{align}
The constant is independent of $L$ by the same finite-cover argument
as in
\cref{proposition:linear-estimate-rescaling-argument}.
The same argument gives the uniform trace estimate
\begin{equation}
\label{equation:metric-u-trace}
    \|{\tr v}_{L^2_{3/2}(\partial X_L)}
    \leq
    c\|{v}_{L^2_2(X_L)}.
\end{equation}

The same covering argument and boundary multiplication estimate \cite[Theorem~4.6.4/2]{Runst2011} give 
\begin{equation}
\label{equation:metric-boundary-product-estimate}
    \|{fg}_{L^2_{3/2}(\partial X_L)}
    \leq
    cL^{-1/2}
    \|{f}_{L^2_{3/2}(\partial X_L)}
    \|{g}_{L^2_{3/2}(\partial X_L)}.
\end{equation}

We obtain
\begin{align*}
&\quad
\|{
B_{\mathscr{N}}(\psi_1,u_1)
-
B_{\mathscr{N}}(\psi_2,u_2)
}_{
L^2_{5/2}(\partial X_L)
\oplus
L^2_{3/2}(\partial X_L)
}
\\
&\leq
cL^{-1/2}
\left[
L\|{\psi_1}_{L^2_3(X_L)}
\|{u_1-u_2}_{L^2_2(X_L)}
+
L\|{\psi_1-\psi_2}_{L^2_3(X_L)}
\|{u_2}_{L^2_2(X_L)}
\right]
\\
&\leq
cLd,
\end{align*}
where we used
\cref{equation:metric-normal-trace}, \cref{equation:metric-u-trace}, \cref{equation:metric-boundary-product-estimate} and the assumptions of the
proposition.

Combining the interior and boundary estimates proves the claim.
\end{proof}

\subsection{An application of the contraction mapping theorem}

In this section we will conclude the proof of
\cref{theorem:metric-main-uniqueness}, and this will be done in analogy
with \cref{subsection:scalar-contraction-mapping-theorem}.
We will show that the map
$-T(\mathscr{N},B_{\mathscr{N}})$ is a contraction on a small ball,
showing that $\psi=0$ and $u=0$ is the only solution to the set of
equations \cref{equation:K-eqn} and \cref{equation:A-eqn}
\emph{near zero}.
In other words, a solution near $\phi=\phi_0$ and $A=0$ must be
independent of the $z$-direction.

We take one last step to show that a solution that is independent of the
$z$-direction must be the homogeneous black string.
Taking both together, this implies that the unique solution to
\cref{equation:A-K-eqn-without-derivative} near $\phi=\phi_0$ and
$A=0$ is the homogeneous black string.

We begin by proving this last step, namely that a solution that is
independent of the $z$-direction must be the homogeneous black string:

\begin{proposition}
    \label{proposition:metric-independent-of-z-trivial}
    Let $\phi \in L^2_4(X_L), A \in L^2_3(X_L)$ be
    $\SO(3)$-invariant solutions to
    \cref{equation:A-K-eqn-without-derivative} that are independent of
    the $z$-direction and satisfying the boundary conditions
    \cref{equation:boundary-conditions-homog}.
    Then $\phi=\phi_0$ and $A=0$.
\end{proposition}

\begin{proof}
    This follows from \cite[Section III.D]{Albertini2024}.
\end{proof}

We are now ready to complete the proof of
\cref{theorem:metric-main-uniqueness}.
By virtue of the previous proposition, it suffices to show that
$\phi_z=0$ and $A_z=0$ for a solution $\phi,A$ near the homogeneous
black string.

\begin{proof}[Proof of \cref{theorem:metric-main-uniqueness}]
    Let $\phi,A$ be a solution satisfying the hypotheses of
    \cref{theorem:metric-main-uniqueness}, and write
    $\psi:=\phi_z$ and $u:=A_z$ so that $(\psi,u)
        \in
        L^2_{3,\mathrm{mean}-0}(X_L)
        \oplus
        L^2_{2,\mathrm{mean}-0}(X_L)$.
    As in the proof of \cref{proposition:metric-linear} we find
\[
    \|{\bar\phi-\phi_0}_{L^2_4(X_L)}
    \leq
    \|{\phi-\phi_0}_{L^2_4(X_L)},
    \qquad
    \|{\bar A}_{L^2_3(X_L)}
    \leq
    \|{A}_{L^2_3(X_L)}.
\]
Thus, by choosing the constant $c$ in the statement of the theorem
sufficiently small, independently of the solution, the zero modes
satisfy the assumptions of \cref{proposition:metric-linear} and
\cref{proposition:metric-nonlinear-estimate}.
    After making $c$ smaller if necessary, the smallness
    assumptions in the theorem also give
    \[
        \|{(\psi,u)}_{L^2_3(X_L)\oplus L^2_2(X_L)}
        \leq
        L^{1/2}.
    \]

    Let $T=T_{\bar\phi,\bar A,L}$ be the left-inverse of
    $(\mathscr{L},B_{\mathscr{L}})$ from
    \cref{corollary:left-inverse-metric}. We consider the closed ball
    \[
        B
        :=
        \left\{
        (\eta,\xi)
        \in
        L^2_{3,\mathrm{mean}-0}(X_L)
        \oplus
        L^2_{2,\mathrm{mean}-0}(X_L):
        \|{(\eta,\xi)}_{L^2_3(X_L)\oplus L^2_2(X_L)}
        \leq
        L^{1/2}
        \right\}
    \]
    and define
    \[
        \mathcal{F}(\eta,\xi)
        :=
        -T\left(
        \mathscr{N}(\eta,\xi),
        B_{\mathscr{N}}(\eta,\xi)
        \right).
    \]

    We first check that $\mathcal{F}$ maps $B$ to itself. Since
    $\mathscr{N}(0,0)=0$ and $B_{\mathscr{N}}(0,0)=0$, for
    $(\eta,\xi)\in B$ we have
    \begin{align*}
        \|{\mathcal{F}(\eta,\xi)}_
        {L^2_3(X_L)\oplus L^2_2(X_L)}
        &\leq
        C
        \|{
        (\mathscr{N},B_{\mathscr{N}})(\eta,\xi)
        -
        (\mathscr{N},B_{\mathscr{N}})(0,0)
        }_{
        L^2_1(X_L)\oplus L^2(X_L)
        \oplus L^2_{5/2}(\partial X_L)
        \oplus L^2_{3/2}(\partial X_L)
        }
        \\
        &\leq
        CL
        \|{(\eta,\xi)}_{L^2_3(X_L)\oplus L^2_2(X_L)}
        \leq
        CL^{3/2},
    \end{align*}
    where we used \cref{corollary:left-inverse-metric} and
    \cref{proposition:metric-nonlinear-estimate}.
    For sufficiently small $L$, this is bounded by $L^{1/2}$.

    Similarly, if $(\eta_i,\xi_i)\in B$ for $i\in\{1,2\}$, then
    \begin{align*}
        \|{
        \mathcal{F}(\eta_1,\xi_1)
        -
        \mathcal{F}(\eta_2,\xi_2)
        }_{L^2_3(X_L)\oplus L^2_2(X_L)}
        &\leq
        CL
        \|{
        (\eta_1,\xi_1)-(\eta_2,\xi_2)
        }_{L^2_3(X_L)\oplus L^2_2(X_L)}.
    \end{align*}
    Thus, after making $L_0$ smaller if necessary, $\mathcal{F}$ is a
    contraction on $B$. By the Banach fixed point theorem, $(0,0)$ is
    its unique fixed point in $B$.
    Differentiating the equations and the boundary condition satisfied
    by $\phi,A$ gives
    \[
        (\mathscr{L},B_{\mathscr{L}})(\psi,u)
        +
        (\mathscr{N},B_{\mathscr{N}})(\psi,u)
        =
        0.
    \]
    Applying the left-inverse $T$ and using
    $T(\mathscr{L},B_{\mathscr{L}})=\Id$ shows that
    \[
        (\psi,u)
        =
        -T\left(
        \mathscr{N}(\psi,u),
        B_{\mathscr{N}}(\psi,u)
        \right)
        =
        \mathcal{F}(\psi,u).
    \]
    Hence $(\psi,u)$ is a fixed point of $\mathcal{F}$ in $B$.
    By uniqueness of the fixed point, $(\psi,u)=(0,0)$, so
    $\phi=\bar\phi$ and $A=\bar A$.

    Finally, \cref{proposition:metric-independent-of-z-trivial} gives
    $\bar\phi=\phi_0$ and $\bar A=0$, which completes the proof.
\end{proof}
\section*{Acknowledgments.}
The authors thank Toby Wiseman for helpful conversations.
E.A. thanks Veronika Hubeny, James Lucietti and Hari Kunduri for valuable feedback after her seminars.
\section*{Funding}
D.P. was supported by the Eric and Wendy Schmidt AI in Science Postdoctoral Fellowship, funded by Schmidt Sciences. E.A. was supported by the STFC Consolidated Grant ST/W507519/1. E.A. acknowledges support from the European Union’s
Horizon ERC Synergy Grant “Making Sense of the Un-expected in the Gravitational-Wave Sky” (Grant No.
GWSky–101167314). 
\section*{Competing interests statement}
The authors have no relevant financial or non-financial interests to disclose.
\section*{Data availability statement}
The authors confirm that the data supporting the findings of this study are available within the article.

\bibliographystyle{unsrt}
\bibliography{library}

\noindent{\small\sc SISSA, Via Bonomea 265, 34136 Trieste, Italy \& INFN Sezione di Trieste 

\noindent E-mail: {\tt \href{mailto:ealberti@sissa.it}{ealberti@sissa.it}}

\vskip 8pt

\noindent{\small\sc Imperial College London, Department of Mathematics, 180 Queen's Gate, South Kensington, London SW7 2RH, the United Kingdom 

\noindent E-mail: {\tt \href{mailto:d.platt@imperial.ac.uk}{d.platt@imperial.ac.uk}}

  \end{document}